\documentclass[10pt]{article}

\usepackage{array}

\usepackage{microtype}
\usepackage{graphicx}
\usepackage{subfigure}
\usepackage{multirow}
\usepackage{booktabs} % for professional tables

\usepackage{hyperref}

\usepackage[accepted]{icml2025}

\usepackage{amsmath}
\usepackage{amssymb}
\usepackage{mathtools}
\usepackage{amsthm}

\theoremstyle{plain}
\newtheorem{theorem}{Theorem}[section]

\theoremstyle{definition}

\theoremstyle{remark}

\usepackage[textsize=tiny]{todonotes}

\graphicspath{{figures/}}

\icmltitlerunning{DawnPiper: A Memory-scablable Pipeline Parallel Training Framework}

\begin{document}

\twocolumn[
\icmltitle{DawnPiper: A Memory-scablable Pipeline Parallel Training Framework}

% It is OKAY to include author information, even for blind
% submissions: the style file will automatically remove it for you
% unless you've provided the [accepted] option to the icml2025
% package.

% List of affiliations: The first argument should be a (short)
% identifier you will use later to specify author affiliations
% Academic affiliations should list Department, University, City, Region, Country
% Industry affiliations should list Company, City, Region, Country

% You can specify symbols, otherwise they are numbered in order.
% Ideally, you should not use this facility. Affiliations will be numbered
% in order of appearance and this is the preferred way.
\icmlsetsymbol{equal}{*}

\begin{icmlauthorlist}
\icmlauthor{Xuan Peng}{hust}
\icmlauthor{Xuanhua Shi}{hust}
\icmlauthor{Haolin Zhang}{hust}
\icmlauthor{Yunfei Zhao}{hust}
\icmlauthor{Xuehai Qian}{thu}
\end{icmlauthorlist}

\icmlaffiliation{hust}{The National Engineering Research Center for Big Data Technology and System, Services Computing Technology and System Lab, Huazhong University of Science and Technology, Wuhan, 430074, China.}
\icmlaffiliation{thu}{Tsinghua University}

\icmlcorrespondingauthor{Xuanhua Shi}{xhshi@hust.edu.cn}

% You may provide any keywords that you
% find helpful for describing your paper; these are used to populate
% the "keywords" metadata in the PDF but will not be shown in the document
\icmlkeywords{GPU, Deep Learning Training, Pipeline Parallelism, Memory Management}

\vskip 0.3in
]

\printAffiliationsAndNotice{}  % leave blank if no need to mention equal contribution
%\printAffiliationsAndNotice{\icmlEqualContribution} % otherwise use the standard text.
%

%!TEX root = main.tex
\begin{abstract}
Pipeline parallelism is a crucial paradigm for large-scale model training.
However, imbalances in memory footprint across stages can lead to significant GPU memory wastage,
limiting the model sizes that pipeline parallelism can effectively support.
In this paper, we introduce DawnPiper, a memory-scalable pipeline parallel training framework. 
Firstly, we develop a DL compilation-based profiling method that transforms the model into a fine-grained computation graph.
This refinement gives us a finer granularity of model partitioning and memory optimization
while facilitating automatic code generation.
Based on observed memory usage characteristics,
we derive a performance-optimal theorem for pipeline parallel partitioning
that substantially reduces the partition search space.
% we derive a pipeline parallelism partition theorem that effectively reduces the partition search space. 
Secondly, we propose a binary pipeline partitioning algorithm and
utilize a cost-model based memory optimization approach to
efficiently identify nearly optimal pipeline parallel strategy.
DawnPiper achieves up to a 4$\times$ and 11$\times$ increase in trainable maximum batch size compared to vPipe and PipeDream, respectively,
and provides up to a 1.5$\times$ performance speedup compared to vPipe.
\end{abstract}

%!TEX root = main.tex
\section{Introduction}
Deep neural networks (DNNs) have achieved remarkable success across various fields,
including computer vision, autonomous driving, and particularly in natural language processing (NLP),
as demonstrated by models like GPT-4~\cite{achiam2023gpt} and Llama3~\cite{dubey2024llama}.
The impressive performance of DNNs stems from advancements in their architecture
and the growing size of their parameters.
However, there is a significant disparity between the exponential growth in model parameters
and the slower increase in GPU memory capacity.
As a result, training such large models exceeds the memory limitations of a single GPU.
Consequently, training DNNs across multiple GPUs has become increasingly
prevalent in both academia~\cite{miao2023hetu,zhengAlpaAutomatingInter2022} and industry~\cite{abadiTensorFlowSystemLargeScale2016,paszkePytorchImperativeStyle2019}.

% There are three primary parallelism methods for training large-scale DNNs:
% \emph{Data parallelism (DP)}, \emph{Tensor parallelism (TP)}, and \emph{Pipeline parallelism (PP)}.
% \emph{DP} divides the training data into multiple shards,
% distributing them across different devices, with each device storing a complete copy of the model parameters.
% During each training iteration, all devices must communicate to synchronize and update the model parameters.
% \emph{TP} distributes computations across multiple GPUs by partitioning tensors along specific dimensions,
% enabling simultaneous processing of different parts of the neural network.
% While both \emph{DP} and \emph{TP} effectively distribute memory and computation across devices,
% \emph{DP} faces significant communication overhead as the model size increases,
% while \emph{TP} struggles with synchronization delays due to inter-device communication.

% In contrast, \emph{pipeline parallelism} aims to parallelize computations across the layers of DNNs.
% It splits DNNs into multiple layer blocks and divides the input data of a batch into smaller micro-batches,
% enabling pipelined execution across different stages.
% Communication occurs only between adjacent stages and can be overlapped with ongoing computations,
% leading to lower communication volume compared to \emph{DP} and lower bandwidth requirements compared to \emph{TP}.
%Nonethless, achieving both high computation efficiency and memory utilization remains a significant challenge in \emph{PP}.

\emph{Pipeline parallelism} aims to parallelize computations across the layers of DNNs.
It splits DNNs into multiple layer blocks and divides the input data of a batch into smaller micro-batches,
enabling pipelined execution across different stages.
Communication occurs only between adjacent stages and can be overlapped with ongoing computations,
leading to lower communication volume compared to \emph{data parallelism} and
lower bandwidth requirements compared to \emph{tensor parallelism}.

Current pipeline parallelism frameworks can be classified into two categories:
\emph{Synchronous Pipeline Parallelism (SPP)} and \emph{Asynchronous Pipeline Parallelism (APP)}.
A notable example of \emph{SPP} is GPipe~\cite{huangGpipeEfficientTraining2019},
which is shown as Figure~\ref{subfig:gpipe}.
Each small rectangle represents the
forward and backward computations of a micro-batch, with
the number inside indicating the micro-batch index.
% The computation of the next mini-batch does not begin until the backward propagation and parameter updates
% of the current mini-batch are complete in \emph{SPP}.
% The most representative works of \emph{SPP} and \emph{APP} are GPipe and PipeDream, respectively,
% which are shown as Figure~\ref{fig:pp-exec}.
However, \emph{SPP} suffers from
low GPU utilization due to a synchronization barrier before
the next mini-batch can start, leading to pipeline bubbles (idle GPU time).% shown as the blank areas in Figure~\ref{subfig:gpipe}.
To tackle this problem, PipeDream~\cite{narayananPipeDreamGeneralizedPipeline2019}
introduces a \emph{1F1B} scheduling mechanism to enhance GPU utilization, which is shown as Figure~\ref{subfig:pipedream}.
PipeDream begins with a warm-up phase to initialize the pipeline,
requiring the first $(\ell-x)$ micro-batches in stage $x$ to be fed into the pipeline before the 1F1B scheduling can begin,
in which $\ell$ represents the number of pipeline stages.
To ensure parameter consistency during both the forward and backward passes within the same micro-batch,
multiple versions of parameters must be maintained.
Specifically, $(\ell-x)$ replicas of parameters are stored at stage $x$,
meaning earlier stages require additional GPU memory to hold these replicas.
\begin{figure}[htb]
  \centering
  % \hspace{-0.7cm}
  \subfigure[Synchronous Pipeline Parallel — GPipe]{
    \centering
    \label{subfig:gpipe}
    \includegraphics[width=1.0\linewidth]{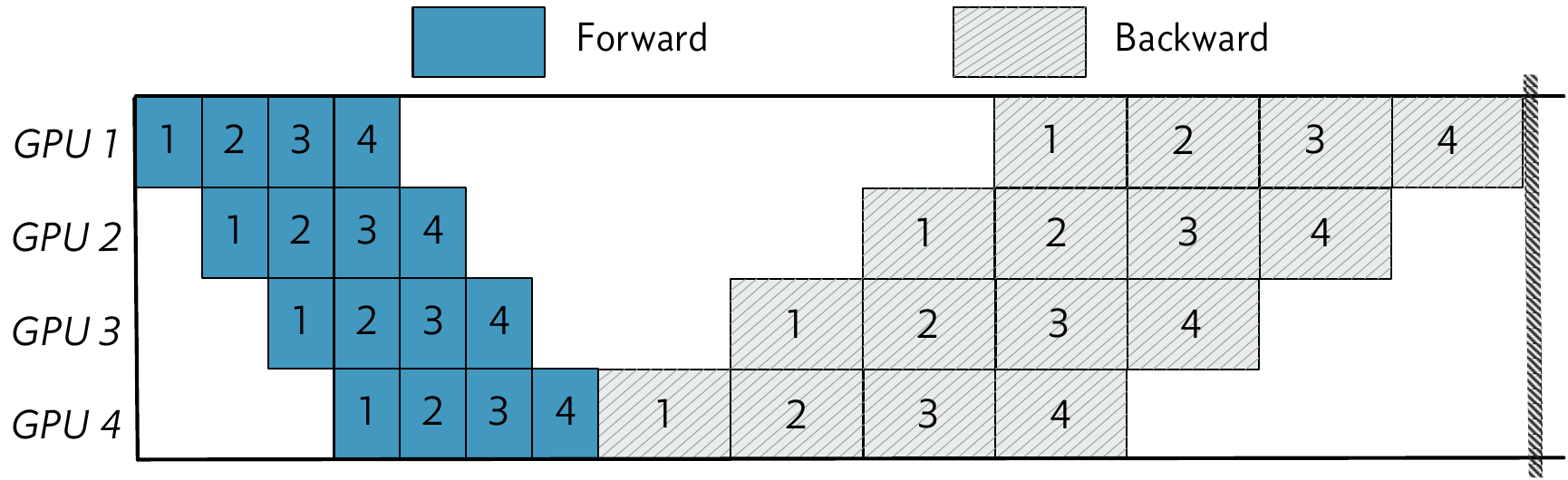}
  }
  \subfigure[Asynchronous Pipeline Parallel — PipeDream]{
    \centering
    \label{subfig:pipedream}
    \includegraphics[width=1.0\linewidth]{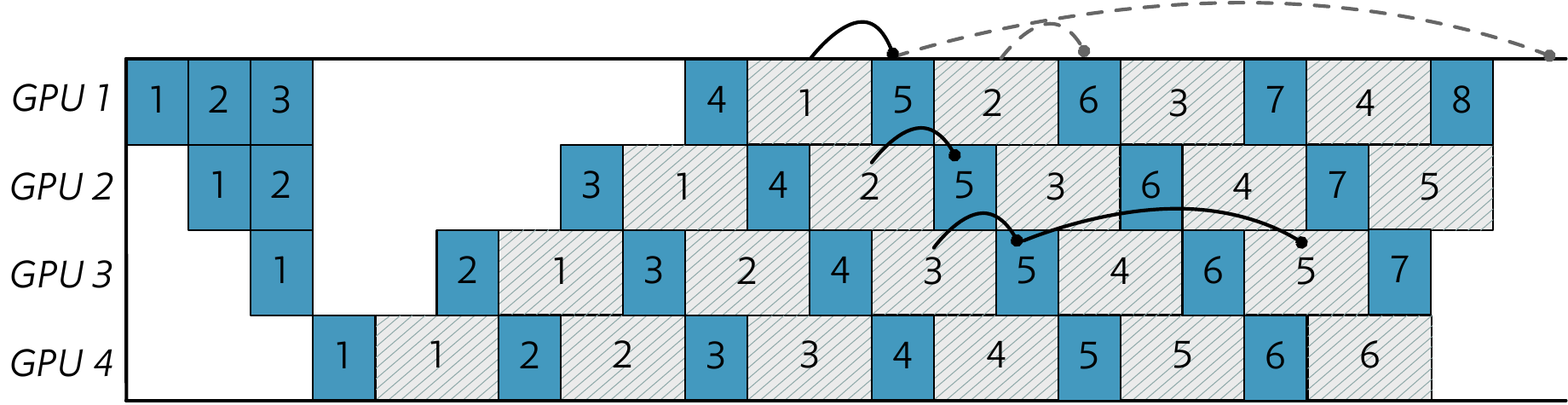}
  }
  \caption{Synchronous and Asynchronous Pipeline Parallel}
  \label{fig:pp-exec}
\end{figure}
% \emph{SPP} typically suffers from pipeline bubble overhead,
% while \emph{APP} requires storing multiple versions of model parameters
% to maintain consistency within the same micro-batch.

Both approaches aim to balance the execution time across stages,
as any lagging (straggler) stage causes delays in the entire pipeline,
slowing down the overall training performance.
However, focusing solely on balancing computation time often overlooks memory usage at different stages.
We have observed over 40\% GPU memory waste in existing pipeline parallel training frameworks.
This waste stems from two factors: 1) the varying ratio of computation time to memory usage across different types of layers,
and 2) the inconsistency in the number of parameter versions, activation memory, and gradients required at each stage,
further exacerbating memory imbalance.
This imbalance significantly impacts the scalability of pipeline parallelism with respect to memory.

\emph{Memory swapping}~\cite{rhuVDNNVirtualizedDeep2016,jinLayerCentricMemoryReuse2018,wangSuperneuronsDynamicGPU2018,huangSwapAdvisorPushingDeep2020,renSentinelEfficientTensor2021}
and \emph{recomputation}~\cite{chenTrainingDeepNets2016,kirisameDynamicTensorRematerialization,heHOMEHolisticGPU2022,korthikantiReducingActivationRecomputation2023}
are two popular techniques for reducing memory usage during DNN training.
\emph{Swapping} leverages CPU memory as an external buffer to extend GPU memory capacity,
while \emph{recomputation} discards activation memory after forward propagation,
regenerating it later through redundant computation.
Both methods influence the memory footprint and execution time of a DNN,
creating a vast and complex search space when combined with pipeline parallel partitioning.
This combined optimization is an NP-hard problem due to its exponential search complexity.
All prior works address pipeline parallel partitioning and memory optimization at a coarse granularity,
typically at the layer or layer block level.
While this approach simplifies partitioning, it limits the optimization space and increases overhead.
In contrast, finer-grained, \emph{tensor-level} optimization methods have proven
more efficient in prior memory optimization studies~\cite{pengCapuchinTensorbasedGPU2020,zhangTENSILETensorGranularity2022,nieTSPLITFinegrainedGPU2022}.
% Additionally, current methods rewrite the user-submitted model code at the layer (block)
% level to enable code generation for each stage.
% However, this approach lacks generalizability in production environments,
% particularly when user models involve custom modules or when privacy concerns
% prevent third-party developers from accessing the model structure.

In this paper, we propose DawnPiper, a memory-scalable pipeline parallelism partitioning method
designed to enhance both pipeline training performance and GPU memory utilization.
Firstly, we develop a DL compilation-based profiling technique
to generate a fine-grained computation graph for the model.
This expanded partition space allows for precise adjustments to the computation time and memory footprint of each stage.
Additionally, the module code for each stage can be automatically generated via DL compilation.
Profiling results from typical DNN training reveal a key memory usage characteristic:
the majority (over 90\%) of activation and operation memory consumption is relatively small (around 100 MB).
Building on this, we derive a performance-optimal theorem in pipeline parallel partitioning that defines
% Building on this, we derive a pipeline parallel partitioning theorem that defines
the partition range between compute-balanced and memory-balanced positions when dividing adjacent pipeline stages.

% Secondly, building on the proposed pipeline partition theorem,
% we introduce a binary pipeline parallel partition algorithm.
Secondly, we introduce a binary pipeline parallel partition algorithm based on the theorem.
Through starting from the middle stage,
we treat the left and right model segments as two adjacent stages
and recursively explore potential partition positions between the compute-balanced and memory-balanced positions.
For each partition plan, we apply a cost-model-based memory optimization method,
as proposed by Capuchin~\cite{pengCapuchinTensorbasedGPU2020}, to minimize memory usage within GPU capacity in linear time.
Then we can record the computation time of the longest execution stage.
By evaluating all partition plans, the strategy that results in the shortest execution time for the longest stage
is identified as the (nearly) optimal partition solution.

We have prototyped DawnPiper on PyTorch~\cite{paszkePytorchImperativeStyle2019}.
Our evaluations demonstrate that DawnPiper achieves up to a 4$\times$ and 11$\times$
increase in maximum trainable batch size compared to
vPipe~\cite{zhaoVPipeVirtualizedAcceleration2022} and PipeDream~\cite{narayananPipeDreamGeneralizedPipeline2019},
respectively, and offers up to a 1.5$\times$ performance speedup compared to vPipe.

\section{Related Work}
\label{sec:related}
\subsection{Data Parallelism}
Data parallelism is a widely used parallel strategy
that divides the training data into multiple shards
and distributes them across different devices.
As model parameters and optimizer states increasingly consume
a significant portion of memory during NLP model training,
Ren et al. introduced the Zero Redundancy Optimizer (ZeRO)~\cite{rajbhandariZeROMemoryOptimizations2020}
to evenly distribute this memory across devices.
ZeRO-Offload~\cite{renZerooffloadDemocratizingBillionScale2021} advances this approach
by offloading all parameters and optimizer computations to the CPU,
and redesigns the CPU optimizer computation method to accelerate the process.
ZeRO-Infinity~\cite{rajbhandariZeROinfinityBreakingGPU2021} builds on ZeRO-Offload
by further optimizing memory usage through offloading both the optimizer and GPU memory to the CPU,
non-volatile memory (NVM),
and other external storage.
% thereby expanding available GPU memory.
PatrickStar~\cite{fangParallelTrainingPreTrained2022} proposed organizing memory in blocks
and using these blocks as communication units to enhance communication bandwidth utilization.
%In general, data parallelism is orthogonal to our work but can be used in conjunction for model training.

\subsection{Pipeline and Hybrid Parallelism}
GPipe~\cite{huangGpipeEfficientTraining2019} and PipeDream~\cite{narayananPipeDreamGeneralizedPipeline2019}
are pioneering works in synchronous and asynchronous pipeline parallelism, respectively.
DAPPLE~\cite{fanDAPPLEPipelinedData2021} introduces the 1F1B computation scheduling into synchronous pipeline parallelism
by dividing the batch into smaller micro-batches during the initial warm-up phase.
It reduces pipeline bubbles at the cost of requiring a sufficient number of micro-batches,
which is not always feasible.
Chimera~\cite{liChimeraEfficientlyTraining2021} proposes a bidirectional synchronous pipeline scheduling method,
reducing pipeline bubbles by enabling a GPU to handle both forward and backward computations of different stages.
Hanayo~\cite{liuHanayoHarnessingWavelike2023} combines DAPPLE and Chimera’s approaches,
introducing a wave-based computation scheduling method.
Narayanan et al.~\cite{narayananEfficientLargeScaleLanguage2021} mitigate pipeline bubbles by
allowing a stage’s GPU to store more model partitions, based on DAPPLE’s scheduling,
though this increases communication overhead.
Qi et al.~\cite{qiZeroBubbleAlmost} achieves nearly bubble-free
pipeline parallelism by re-scheduling backward computation.
BPipe~\cite{kimBPipeMemoryBalancedPipeline2023} balances GPU memory usage
between stages by asynchronously exchanging activations between GPUs.
vPipe~\cite{zhaoVPipeVirtualizedAcceleration2022} uses an iterative algorithm to
dynamically balance model partitioning, memory swapping, and recomputation, employing the Kernighan-Lin algorithm~\cite{kernighanEfficientHeuristicProcedure1970}.
AdaPipe~\cite{sunAdaPipeOptimizingPipeline2024} introduces a two-step
dynamic programming approach for pipeline partitioning and recomputation.
Pipeline parallelism is often used in conjunction with tensor parallelism,
as seen in hybrid parallelism frameworks like Flexflow~\cite{jiaDataModelParallelism2019} and Alpa~\cite{zhengAlpaAutomatingInter2022}.
We plan to explore the extension of DawnPiper to hybrid parallelism in future work.

%\input{background}

% \input{observation}

%!TEX root = main.tex
% \section{Motivations and Opportunities}
\section{Challeges and Motivations}
% \section{Motivations}
\label{sec:cam}
\subsection{Memory Imbalance in Pipeline Parallelism}
\label{subsec:mipp}
Pipeline parallelism is designed to distribute large model training across multiple devices
when a single device's memory cannot accommodate it.
However, significant memory inefficiency arises due to imbalanced memory usage across pipeline stages,
which contradicts its intended purpose.
To highlight this imbalance,
we conducted a benchmark test on the peak memory usage of each stage in GPipe and PipeDream.
The test was performed on a server with 8 A100 GPUs (40 GB each),
with pipeline stages set to 4 and 8 for GPipe and PipeDream, respectively.
\begin{figure}[t]
  \centering
  \begin{minipage}[t]{0.49\linewidth}
    \centering
    \includegraphics[width=0.95\linewidth]{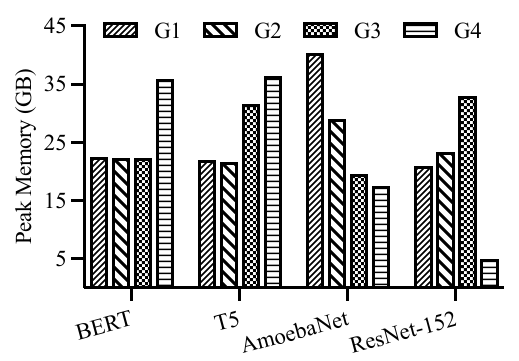}
    \caption{Peak Memory Usage on GPipe (4 GPUs)}
    \label{fig:gpipe-4gpus}
  \end{minipage}
  \begin{minipage}[t]{0.49\linewidth}
    \centering
    \includegraphics[width=0.95\linewidth]{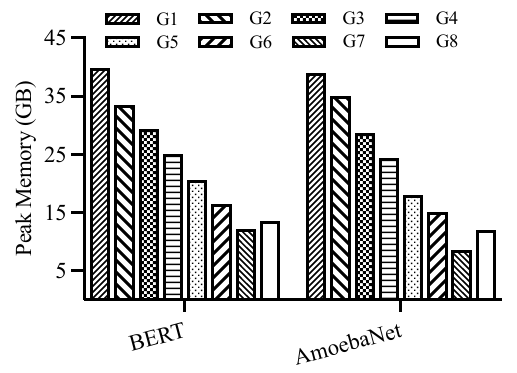}
    \caption{Peak Memory Usage on PipeDream (8 GPUs)}
    \label{fig:pd-8gpus}
  \end{minipage}
\end{figure}
% \begin{figure*}[t]
%   \centering
%   \begin{minipage}[t]{0.33\linewidth}
%     \centering
%     \includegraphics[width=0.95\linewidth]{GPipe-4GPUs.pdf}
%     \caption{Peak Memory Usage of Each Stage on GPipe (4 GPUs)}
%     \label{fig:gpipe-4gpus}
%   \end{minipage}
%   \begin{minipage}[t]{0.33\linewidth}
%     \centering
%     \includegraphics[width=0.95\linewidth]{PipeDream-8GPUs.pdf}
%     \caption{Peak Memory Usage of Each Stage on PipeDream (8 GPUs)}
%     \label{fig:pd-8gpus}
%   \end{minipage}
%   \begin{minipage}[t]{0.33\linewidth}
%     \centering
%     \includegraphics[width=0.95\linewidth]{normalized-perf.pdf}
%     \caption{Normalized Performance of Two Different Partition Methods}
%     \label{fig:norm-perf}
%   \end{minipage}
% \end{figure*}

Figure~\ref{fig:gpipe-4gpus} and \ref{fig:pd-8gpus} presents the benchmark results
for GPipe and PipeDream, respectively,
where G1 to G8 represent the GPU of different pipeline stages.
In GPipe, memory imbalance primarily arises from the uneven ratio
of computation time to memory usage among different layers in the model.
If the maximum peak memory usage across all stages is used as the baseline for GPU capacity,
this imbalance results in nearly 40\% and 34\% memory waste in ResNet-152 and AmoebaNet, respectively.
In PipeDream, this imbalance is further aggravated
by the varying number of copies of model parameters, activation memory, and gradients needed at each stage.
For example, the memory waste ratio for BERT is around 28\% in GPipe, but in PipeDream, this waste ratio exceeds 40\%.

% Based on the discussion above,
It is evident that focusing on balancing computation across pipeline stages
often leads to imbalanced memory usage, and vice versa.
Additionally, with memory optimization techniques like swapping and recomputation,
a key challenge arises when GPU memory is oversubscribed after partitioning a model for computational balance:
should the model be repartitioned to achieve more balanced memory usage,
or should memory optimization methods be applied to handle stages that exceed GPU capacity?
It’s difficult to determine which approach offers better performance.
See Section~\ref{sec:pcbmbcb} of the appendix for detailed discussion.
\begin{figure}[htb]
  \centering
  \subfigure[Activation memory]{
    \centering
    \includegraphics[width=0.47\linewidth]{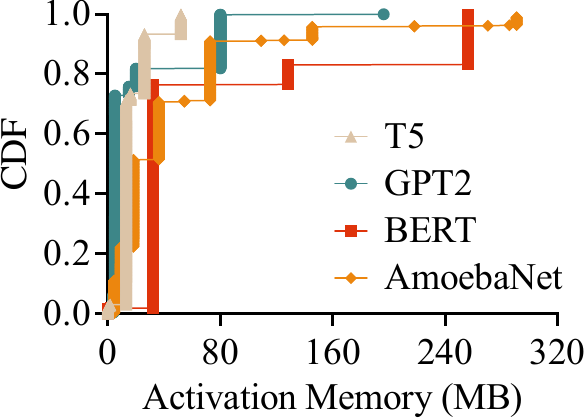}
    \label{fig:act-mem-cdf}
  }
  \subfigure[Consumed Memory]{
    \centering
    \includegraphics[width=0.47\linewidth]{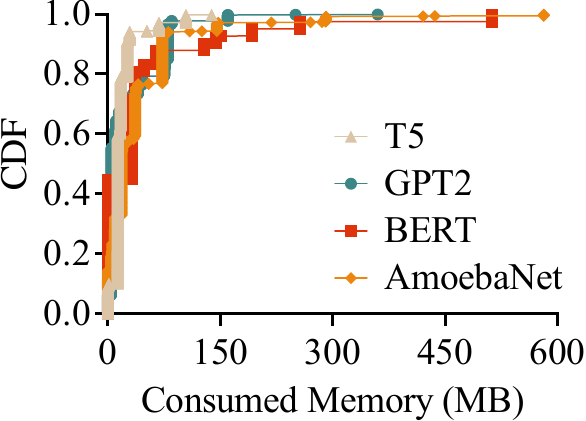}
    \label{fig:cons-mem-cdf}
  }
  \caption{CDF of Node's Activation and Consumed Memory}
  \label{fig:mem-cdf}
\end{figure}

\subsection{Memory Usage Characteristics in DL Training}
\label{sec:mot}
% We observed that DL compilation techniques like XLA~\cite{sabneXlaCompilingMachine2020}
% and TorchDynamo~\cite{anselPyTorchFasterMachine2024},
% widely used in tensor parallelism for fine-grained analysis and automatic code generation,
% have not been applied to pipeline parallelism in previous research.
% Building on the fine-grained analysis provided by DL compilation,
% we identified two key memory usage characteristics during model training that can guide pipeline partitioning.

% First, most activation memory sizes are small,
% making it easier to find low communication volumes between stages for efficient pipeline partitioning.
Based on the fine-grained computation graph obtained through DL compilation,
we profiled the activation and consumed memory sizes for BERT, T5, GPT-2, and AmoebaNet,
with results shown in Figure~\ref{fig:mem-cdf}.
% with results shown in Figure~\ref{fig:act-mem-cdf} and Figure~\ref{fig:cons-mem-cdf}.
% We profiled activation memory sizes for BERT, T5, GPT-2, and AmoebaNet, with results shown in Figure~\ref{fig:act-mem-cdf}.
% In Transformer-based models, the largest activation memory is typically at the output layer,
% which resides in the last stage and does not contribute to inter-stage communication.
% For AmoebaNet, the largest activation memory (about 580 MB) comes from outputs of certain convolutional layers,
% appearing only five times, so they are omitted in the figures. 
It reveals that nearly 90\% of computational nodes have activation memory sizes below 80 MB.
In T5, most activation memory is only around 50 MB.
In BERT, nodes with activation memory no greater than 128 MB account for over 80\% of the total nodes.
Memory consumption is calculated as the sum of memory allocated and released during node execution,
with released memory represented as a negative value.
It can be seen that nearly 90\% of memory consumption per node is under 150 MB.
% The results show that in all four models, nearly 90\% of memory consumption per node is under 150 MB.
This suggests that partition position between stages can be flexibly adjusted between nodes with minimal memory usage,
allowing for small changes in the overall memory footprint of each stage while keeping communication overhead low.
Based on these insights, we propose a partitioning theorem for pipeline parallelism,
which will be discussed in Section~\ref{subsec:ptpt}.
% which will be discussed in detail in the next section.
% Even as batch size increases with pipeline partitioning, the communication overhead remains manageable.
% For example, with four pipeline stages, a 4$\times$ transfer of 128 MB (512 MB total) over
% a PCIe 3.0 $\times$16 link would take less than 40 ms, which is significantly shorter than the computation time.

% Secondly, the distribution of absolute memory consumption is shown in Figure~\ref{fig:cons-mem-cdf}.
% Memory consumption here is calculated as the sum of memory allocated and released during node execution,
% with released memory represented as a negative value.
% The results show that in all four models, nearly 90\% of memory consumption per node is under 150 MB.
% This suggests that partition points between stages can be flexibly adjusted between nodes with minimal memory usage,
% allowing for small changes in the overall memory footprint of each stage while keeping communication overhead low.

% Based on these insights, we propose a partitioning theorem for pipeline parallelism,
% which will be discussed in detail in the next section.

%!TEX root = main.tex
\section{Design of DawnPiper}
\label{sec:design}
The overall architecture of DawnPiper is illustrated in Figure~\ref{fig:sys_arch},
comprising three key modules responsible for pipeline parallelism partitioning.
When a model is submitted to DawnPiper, the \emph{compiler} first generates a fine-grained
computation graph through DL compilation.
Next, execution metadata for each node is gathered via profiling.
Finally, the partitioning and memory optimization policy are determined based on the profiling results.
\begin{figure*}[t]
  \centering
  \includegraphics[width=0.95\linewidth]{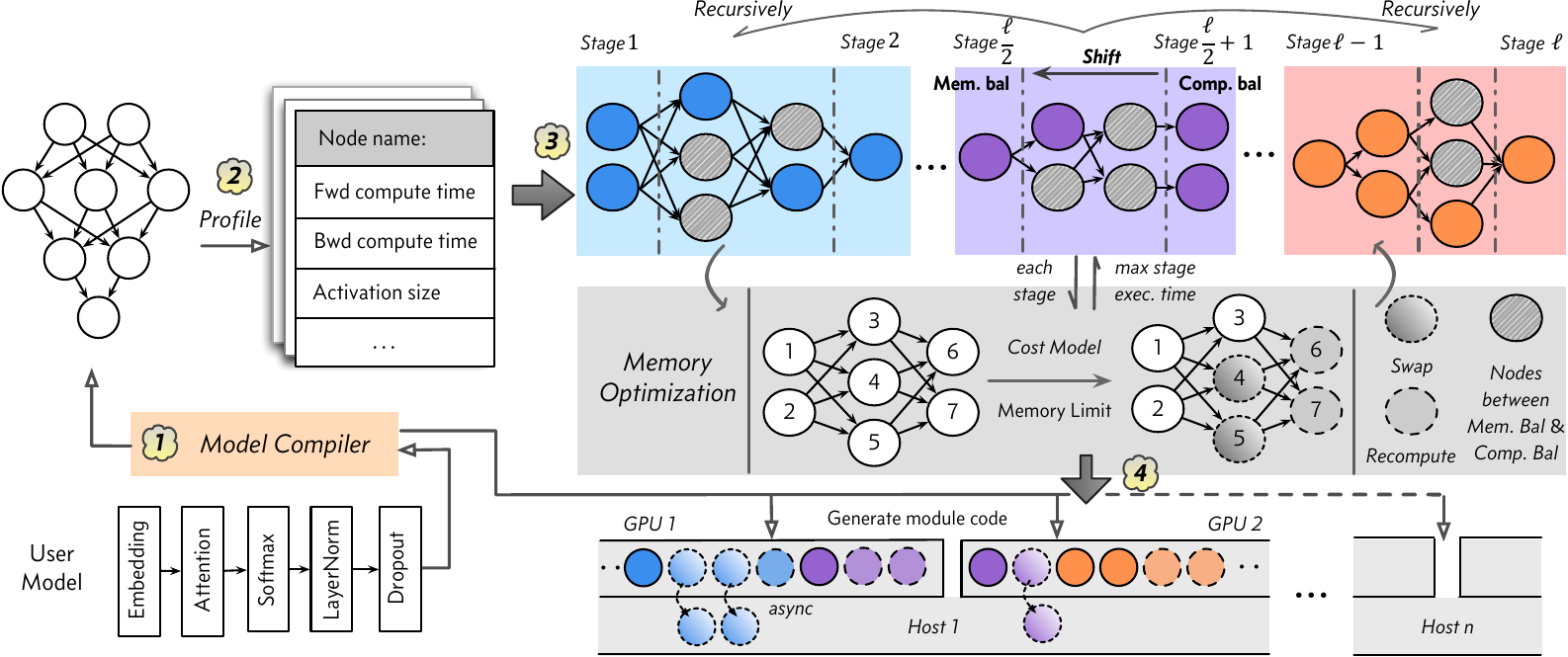}
  \caption{DawnPiper System Architecture}
  \label{fig:sys_arch}
\end{figure*}

\subsection{Pipeline Terminology and Partition Theorem}
\label{subsec:ptpt}
The number of pipeline stages is denoted $\ell$.
For a pipeline stage $x$, where $x \in [1, \ell]$,
the set of computation nodes it contains is represented as $\mathcal{N}_x$.
The computation time for this stage is determined by summing the
forward and backward computation times of all nodes in $\mathcal{N}_x$, as shown in Equation~(\ref{equ:tx}).
\begin{align}
  T_x=\sum_{n \in \mathcal{N}_x}\left(t_f^n+t_b^n\right)
\label{equ:tx}
\end{align}

The peak memory usage for a stage is calculated based on the nodes within the stage and the pipeline scheduling method.
The goal of optimal partitioning is to balance the computation times across all stages,
ensuring they are as close as possible.
Since inter-stage communication is minimal and can be overlapped with computation,
the objective is to minimize the execution time of the longest stage, i.e.,
the stage with the maximum computation time across all partitioning strategies.
This optimization is subject to the constraint that the peak GPU memory usage of
all stages remains within the GPU’s capacity.
In other words, the goal is to $minimize\ \max_{1 \leq x \leq \ell} T_x$.

Given the fine-grained computation graph generated through DL compilation,
the search space for pipeline partitioning grows exponentially
when considering all possible partition positions combined with different memory optimization strategies.
This makes the search space extremely large.
To address this, it is crucial to eliminate impractical partition strategies and reduce the overall search space.
We propose a partition theorem that leverages the memory usage characteristics
to effectively limit the partition space.

\begin{theorem}
\label{thm:ppp}
For a model that needs to be divided into two pipeline stages,
the optimal partition point will lie within the closed interval between the
compute-balanced and memory-balanced positions, provided the following three conditions are met:
1) The computation time and memory usage during forward propagation exhibit
a monotonically increasing trend as operations are progressively scheduled;
2) The inter-stage communication time is consistently less than the computation time of each stage;
3) Memory optimization opportunities are evenly distributed across the model.
\end{theorem}

\begin{proof}
The compute-balanced and memory-balanced
partition positions are denoted as $\rho_{cb}$ and $\rho_{mb}$, respectively.
In the 1F1B computation scheduling of asynchronous pipelines parallelism,
the earlier stages generally need to store more activation and parameter memory,
resulting in $\rho_{cb}$ typically being positioned further to the right than $\rho_{mb}$
(the reverse case can be proven similarly).
When the partition point is shifted to the right from $\rho_{cb}$,
more computation nodes are allocated to the first stage.
Due to the monotonic increase in both computation time and memory usage,
and given that inter-stage communication overhead is negligible,
the computation time and peak memory usage of the first stage increase,
while those of the second stage decrease.
Since memory optimization opportunities are evenly distributed across the model,
no stage with lower memory usage will require a disproportionately
high memory optimization cost to stay within the GPU memory limit.
Consequently, regardless of whether memory optimization is applied,
the first stage becomes the stage with the longest computation time,
surpassing the longest stage under the compute-balanced partition,
leading to a decline in pipeline training performance as the partition shifts rightward from $\rho_{cb}$.
Similarly, if the partition point is moved leftward from $\rho_{mb}$,
the second stage will become the one with the longest computation time
which will continue to increase as the partition point shifts further left.
\end{proof}

\subsection{Binary Pipeline Parallel Partition}
% If the peak memory usage of all stages under the compute-balanced partition is within the GPU memory capacity,
% this partition strategy already represents the optimal pipeline partition in terms of performance.
% However, if any stage exceeds the GPU memory limit,
% the partition point must either be gradually adjusted from $\rho_{cb}$ toward $\rho_{mb}$
% or memory optimization strategies must be applied to reduce memory usage.
% Both methods alter the execution time of the stages,
% making it difficult to predict which partition and memory optimization strategy will yield the best pipeline performance.
% Consequently, a combinatorial search over these options is necessary.
For a given pipeline partition,
the goal of memory optimization for each stage is to minimize the optimization cost
while ensuring it stays within the GPU memory capacity.
Let $T_x^{moo}$ represent the minimum memory optimization cost time for stage $x$.
The objective then becomes $minimize\ \max_{1 \leq x \leq \ell} (T_x + T_x^{moo})$.

To efficiently handle partitioning for deeper pipelines, we employ a binary partitioning approach.
This method first selects the middle stage of the pipeline,
treating the left and right segments as separate stages,
and applies compute-balanced and memory-balanced partitioning.
See Section~\ref{subsec:mcbpa} of the appendix for details.
This approach allows the shortest computation time for the longest execution stage
on either side to be determined incrementally, level by level.
Once all possible partition positions have been explored at the outermost level,
the optimal partition strategy is identified.
This approach reduces the partitioning algorithm's complexity to $\mathcal{O}(\varphi^{log\ell})$,
where $\varphi$ represents the number of potential partition positions traversed between $\rho_{cb}$ and $\rho_{mb}$.
This number is significantly smaller than the total number of nodes in the model.
% For example, BERT has over 1000 nodes, but the number of nodes within this range is less than 100
% when considering communication optimization simultaneously.
% The details can be seen in Section~\ref{subsec:commopt} of the appendix.
% The overall process for this partitioning method is illustrated in Algorithm~\ref{alg:binary-partition}.
For instance, BERT has over 1,000 nodes,
but fewer than 100 nodes fall within this range when communication optimization is considered simultaneously.
Details of communication optimization are provided in Section B.2 of the appendix.
The complete process for this partitioning method is outlined in Algorithm\ref{alg:binary-partition}.

The \texttt{AdjacentPartition} function in Algorithm~\ref{alg:binary-partition}
outlines the process for identifying the partition and memory optimization strategy
that minimizes the difference in computation time between two adjacent stages.
% The compute-balanced $\rho_{cb}$ and memory-balanced $\rho_{mb}$ points
% for these two stages can be derived by adjusting Algorithm~\ref{alg:mem-ba}.
% Lines 8 describe the memory optimization procedure,
% which determines the policy that results in the shortest execution time within the memory constraints.
% Lines 10-12 allow for early termination of the traversal if no further memory optimization is needed between the two stages.
% The function returns the optimal partition, memory optimization strategy,
% and the corresponding execution time for the longer of the two stages.
The \texttt{BiPar} function serves as the core of the binary partitioning algorithm.
Lines 17-19 call the \texttt{AdjacentPartition} function when only two adjacent stages remain.
Otherwise, the function calculates the compute-balanced and memory-balanced points when partitioning from the midpoint.
Lines 24-25 capture the optimal results for the left and right partitions,
which are recursively refined.
% including the partition points, memory optimization strategies
% and the shortest execution time for the longest stage.
By invoking \texttt{BiPar}$(\mathcal{N}_\mathcal{G}, 1, \ell)$,
the final optimal pipeline partitioning and memory optimization strategy is determined.
Note that while the algorithm assumes the number of pipeline stages to be a power of 2,
it is adaptable to any number of stages.
When the remaining stages are reduced to three,
the algorithm performs a traversal from the compute-balanced and memory-balanced partition positions of the first stage.
% Therefore, the proposed approach is capable of handling any arbitrary number of pipeline stages.
\begin{algorithm}
  \caption{Binary pipeline parallel partition}
  \label{alg:binary-partition}
\begin{algorithmic}[1]
%\SetKwFunction{Fmain}{BinaryPartition}
%\SetKwFunction{Fadj}{AdjacentPartition}
%\SetKwProg{Fn}{Function}{:}{}
\FUNCTION{AdjacentPartition($\mathcal{N},sid$)}
%\Fn{\Fadj{$\mathcal{N},sid$}}
  \STATE $\rho_{cb},\rho_{mb} \leftarrow CompMemBalPartition(\mathcal{N}, sid, \ell)$
  \IF{$\max_{sid \le x \le sid+1} ((\ell-x+1) \times M_x^{\rho_{cb}}) < M$}
      \STATE \textbf{return} $\max_{sid \le x \le sid+1} T_x^{\rho_{cb}}$
  \ENDIF
  \STATE $sorted\_pps \leftarrow IdentifyAndSort(\rho_{cb},\rho_{mb})$
  \FOR{$\rho \leftarrow sorted\_pps$}
      \STATE $mopt, min\_t \leftarrow MemOpt(\mathcal{N}_{sid},\mathcal{N}_{sid+1},M)$
      \STATE $par\_dict[\rho] \leftarrow (min\_t,mopt)$
      % \STATE $mopt_l,min\_t_l \leftarrow MemOpt(\mathcal{N}_{sid},M)$
      % \STATE $mopt_r,min\_t_r \leftarrow MemOpt(\mathcal{N}_{sid+1},M)$
      % \STATE $mopt \leftarrow mopt_l \cup mopt_r$
      % \STATE $par\_dict[\rho] \leftarrow (\max(min\_t_l,min\_t_r),mopt)$
      \IF{$!mopt$}
      \STATE \textbf{break}\;
      \ENDIF
  \ENDFOR
  \STATE \textbf{return} $\min(par\_dict)$
\ENDFUNCTION

\FUNCTION{BiPar($\mathcal{N},\mathcal{L},\mathcal{R}$)}
  \IF{$\mathcal{R}-\mathcal{L}==1$}
    \STATE \textbf{return} $AdjacentPartition(\mathcal{N},\mathcal{L})$
  \ENDIF
  \STATE $mid \leftarrow (\mathcal{L}+\mathcal{R}) \div 2$
  \STATE $\rho_{cb},\rho_{mb} \leftarrow CompMemBalPartition(\mathcal{N}, mid, \ell)$
  \STATE $sorted\_pps \leftarrow IdentifyAndSort(\rho_{cb},\rho_{mb})$
  \FOR{$\rho \leftarrow sorted\_pps$}
    % \STATE $\rho_l,mopt_l,min\_t_l \leftarrow BinaryPartition(\mathcal{N}_{\mathcal{L}},\mathcal{L},mid)$
    % \STATE $\rho_r,mopt_r,min\_t_r \leftarrow BinaryPartition(\mathcal{N}_{\mathcal{R}},mid+1,\mathcal{R})$
    \STATE $mopt_l,min\_t_l \leftarrow BiPar(\mathcal{N}_{\mathcal{L}},\mathcal{L},mid)$
    \STATE $mopt_r,min\_t_r \leftarrow BiPar(\mathcal{N}_{\mathcal{R}},mid+1,\mathcal{R})$
    \STATE $mopt \leftarrow mopt_l \cup mopt_r$
    \STATE $par\_dict[\rho] \leftarrow (\max(min\_t_l,min\_t_r),mopt)$
    \IF{$!mopt$}
      \STATE \textbf{break}\;
    \ENDIF
  \ENDFOR
  \STATE \textbf{return} $\min(par\_dict)$
\ENDFUNCTION
\end{algorithmic}
\end{algorithm}

% \textbf{Communication Optimization:} A key assumption in Theorem~\ref{thm:ppp} is
% that communication overhead should not impact the efficiency of pipeline parallelism.
% Therefore, it is essential to avoid partition points that generate significant
% communication volumes during the traversal within the partition range.
% As discussed in Section~\ref{sec:mot}, the majority of computation nodes have very small activation memory.
% First, we can identify inevitable communication points,
% i.e., nodes with edges extending far outside
% the partition range where data transmission is unavoidable regardless of the partition strategy,
% such as the output of the initial embedding layer in BERT.
% This inevitable communication can be disregarded.
% For the remaining nodes, communication overhead can be minimized by optimizing the partition positions.
% For instance, if activations from multiple nodes at the same depth need to be transmitted to the next stage,
% the algorithm can identify the lowest common leaf node connecting these edges.
% If this leaf node is the sole connection point,
% the partition can be adjusted so that only the activation of this single node is transmitted, instead of multiple nodes.
% Such leaf nodes are typically easy to locate,
% resulting in minimal impact on both the computation time and memory usage of the adjacent stages.

\subsection{Memory Optimization}
The goal of memory optimization is to minimize the execution time of each stage
while adhering to the GPU memory limits for a given partitioning.
This objective is very similar to memory optimization in single-GPU training.
To achieve this, we employ the cost model from Capuchin~\cite{pengCapuchinTensorbasedGPU2020},
which evaluates memory swapping and recomputation using \emph{FreeTime} and \emph{Memory Saving Per Second (MSPS)}.
The model first applies memory swapping until data transfer can no longer overlap with computation.
Then, it selects the method either memory swapping or recomputation that incurs the smallest performance overhead.
In DawnPiper, memory optimization candidates are easily identified through each node's saved tensors list,
and their \emph{FreeTime} and \emph{MSPS} can be calculated from the node’s metadata.

\section{Evaluation}
\label{sec:evaluation}
\subsection{Experimental setup}
\textbf{Server and Workloads.}
The experiments were conducted on a server equipped with
8 NVIDIA Tesla A100 GPUs (40 GB each) connected via PCIe 4.0 ×16.
We evaluate our approach using four state-of-the-art deep learning models:
BERT~\cite{devlinBERTPretrainingDeep2019}, GPT-2~\cite{radfordLanguageModelsAre2019},
T5 (Text-to-Text Transfer Transformer)~\cite{raffelExploringLimitsTransfer2020}, and AmoebaNet~\cite{realRegularizedEvolutionImage2019}.
The four models consist of 340 million, 770 million, 780 million, and 28 million parameters, respectively.
We evaluate two pipeline configurations with 4 stages and 8 stages.

\textbf{Baselines.}
We have set up four baselines to evaluate the memory efficiency and training performance,
including ZeRO~\cite{rajbhandariZeROMemoryOptimizations2020},
torchgpipe~\cite{kimTorchgpipeOntheflyPipeline2020} (a PyTorch reimplementation of GPipe),
PipeDream~\cite{narayananPipeDreamGeneralizedPipeline2019},
and vPipe~\cite{zhaoVPipeVirtualizedAcceleration2022}.
% We have set up four baselines to evaluate the memory efficiency and training performance, which are listed below.
% \begin{itemize}
% \item ZeRO~\cite{rajbhandariZeROMemoryOptimizations2020}: ZeRO is developed by Microsoft
% to eliminate redundant parameter storage across GPUs in data parallel training.
% We include optimization levels ZeRO-2 and ZeRO-3 in our evaluation.
% \item GPipe~\cite{huangGpipeEfficientTraining2019}: GPipe is originally implemented based on Lingvo~\cite{shen2019lingvo}
% we instead use torchgpipe~\cite{kimTorchgpipeOntheflyPipeline2020}, a PyTorch reimplementation of GPipe, for our evaluation.
% \item PipeDream~\cite{narayananPipeDreamGeneralizedPipeline2019}: PipeDream profiles the model at runtime
% and uses dynamic programming to minimize the longest computation time across all stages.
% \item vPipe~\cite{zhaoVPipeVirtualizedAcceleration2022}: vPipe is built on PipeDream’s runtime
% while its partitioning algorithm is not open-sourced.
% Therefore, we implemented it according to the description provided in the corresponding paper of vPipe.
% \end{itemize}

\subsection{Memory Efficiency}
In this section, we assess the maximum trainable batch size
for each method to evaluate GPU memory efficiency.
% The batch size is determined by multiplying the micro-batch size
% by the number of splits in synchronous pipeline parallelism.
% We compare micro-batch sizes in asynchronous pipeline parallel training,
% as it differs from synchronous training,
% where there is a distinct gap between the computations of consecutive batches.

\subsubsection{Synchronous Pipeline Parallel Training}
% We first evaluate the synchronous pipeline parallel training mode.
% In this experiment, 
The number of micro-batches is equal to the number of pipeline stages in \emph{SPP}.
The results are shown in Table~\ref{tab:maxbs-sync},
where the synchronous training modes of vPipe and DawnPiper are denoted as vPipe-S and DPiper-S, respectively.

The results indicate that ZeRO-2 and ZeRO-3 achieve nearly the same maximum batch sizes across the four models.
% In fact, the maximum batch sizes of ZeRO-3 are even smaller for the T5 and AmoebaNet models.
% This is because, as the batch size increases, the additional memory saved
% by ZeRO-3's model parameter optimizations is no longer sufficient to support larger batches,
% as ZeRO-3 also requires extra memory buffers for communication.
For GPipe and vPipe, the maximum batch sizes are generally comparable to ZeRO for BERT, GPT-2, and T5.
However, they are significantly smaller than ZeRO for the AmoebaNet model.
This difference arises because in transformer-based models,
the computation and memory usage of each layer are roughly proportional.
As a result, GPipe and vPipe, which aim to balance computation time across stages,
can also effectively balance memory usage.
In contrast, CNN models like AmoebaNet have layers where computation time and memory usage are often not proportional.
For instance, a convolutional layer may have a long computation time but minimal memory usage,
leading to an uneven memory footprint when balancing computation time across stages.
Consequently, the maximum batch sizes that GPipe and vPipe achieve on AmoebaNet
are only 68\% and 46\% of those achieved by ZeRO when the number of pipeline stages is 4.
These ratios drop further to 58\% and 32\% as the number of pipeline stages increases to 8.
Meanwhile, GPipe supports larger batch sizes than vPipe,
but the result reverses when memory optimization is enabled.
This reversal occurs because the stage with the memory bottleneck shifts when memory optimization is applied.
\begin{table}[htbp]
\caption{Maximum Batch Size in Synchornous Pipeline Parallel}
\label{tab:maxbs-sync}
\centering
  \resizebox{0.5\textwidth}{!}{
      \begin{tabular}{m{2em}|m{2em}|ccccc}
      % \begin{tabular}{l|l|ccccc}
      \toprule
      \#\ GPU & MO & Method & BERT & GPT-2 & T5 & AmoebaNet \\
      \midrule
      \multirow{8}{*}{4} & \multirow{5}{*}{None} & ZeRO-2 & 64    & 8     & 140   & 412 \\
            &       & ZeRO-3 & 64    & 8     & 132   & 404 \\
            &       & GPipe & 52    & 8     & 144   & 280 \\
            &       & vPipe-S & 60    & 8     & 144   & 188 \\
            &       & \textbf{DPiper-S} & \textbf{72} & \textbf{8} & \textbf{160} & \textbf{360} \\
  \cmidrule{2-7}          & R & GPipe & 120   & 28    & 348   & 644 \\
  \cmidrule{2-2}          & \multirow{2}{*}{S+R} & vPipe-S & 128   & 28    & 348   & 720 \\
            &       & \textbf{DPiper-S} & \textbf{160} & \textbf{32} & \textbf{424} & \textbf{1220} \\
      \midrule
      \multirow{8}{*}{8} & \multirow{5}{*}{None} & ZeRO-2 & 128   & 16    & 224   & 824 \\
            &       & ZeRO-3 & 128   & 16    & 208   & 808 \\
            &       & GPipe & 88    & 16    & 256   & 480 \\
            &       & vPipe-S & 88    & 16    & 256   & 264 \\
            &       & \textbf{DPiper-S} & \textbf{136} & \textbf{16} & \textbf{320} & \textbf{560} \\
  \cmidrule{2-7}          & R & GPipe & 208   & 48    & 664   & 1384 \\
  \cmidrule{2-2}          & \multirow{2}{*}{S+R} & vPipe-S & 208   & 56    & 664   & 1424 \\
            &       & \textbf{DPiper-S} & \textbf{280} & \textbf{60} & \textbf{800} & \textbf{2180} \\
    \bottomrule
    \end{tabular}}
\end{table}

DawnPiper achieves the largest batch sizes among all methods for the BERT, GPT-2, and T5 models.
Without memory optimization, it increases the batch size by up to 1.2$\times$ and 1.55$\times$
compared to the second-best method, vPipe, when the pipeline stages are 4 and 8, respectively.
Note that for GPT-2,
the maximum batch sizes remain consistent across all methods
due to the model's substantial memory footprint during training,
where even a slight increase in batch size results in a significant rise in memory usage.
Although DawnPiper refines the pipeline partitioning,
it cannot further increase the batch size under these conditions.
For AmoebaNet, DawnPiper's maximum batch size is slightly smaller than ZeRO’s,
as data parallelism can more effectively divide the memory footprint across GPUs,
which is challenging for pipeline parallelism.
Nonetheless, DawnPiper outperforms GPipe and vPipe with batch size increases of 1.29$\times$ and 1.91$\times$
when the pipeline stages are set to 4.
As the number of stages increases to 8, it still surpasses GPipe and vPipe by up to 1.17$\times$ and 2.12$\times$, respectively.

With memory optimization enabled, DawnPiper continues to achieve the largest batch sizes.
Compared to the second-best method, vPipe, DawnPiper increases the maximum batch size by
an average of 1.33$\times$ and 1.29$\times$ for pipeline stages of 4 and 8, respectively.
Additionally, the results show that the average batch size achieved by DawnPiper
increases by 1.84$\times$ as the number of pipeline stages doubles, demonstrating its strong memory scalability.

% Table generated by Excel2LaTeX from sheet '工作簿2'
\begin{table}[htbp]
\caption{Maximum Batch Size in Asynchronous Pipeline Parallel}
\label{tab:maxbs-async}%
\centering
  \resizebox{0.5\textwidth}{!}{
    \begin{tabular}{m{2em}|m{2em}|m{5em}cccc}
    % \begin{tabular}{l|l|ccccc}
    \toprule
    \#\ GPU & MO & Method & BERT & GPT-2 & T5 & AmoebaNet \\
    \midrule
    \multirow{5}{*}{4} & \multirow{3}{*}{None} & PipeDream & 16    & 2     & 40    & 70 \\
          &       & vPipe-AS & 16    & 2     & 40    & 48 \\
          &       & \textbf{DPiper-AS} & \textbf{32} & \textbf{3} & \textbf{78} & \textbf{150} \\
\cmidrule{2-7}          & \multirow{2}{*}{S+R} & vPipe-AS & 46    & 8     & 110   & 300 \\
          &       & \textbf{DPiper-AS} & \textbf{84} & \textbf{10} & \textbf{192} & \textbf{480} \\
    \midrule
    \multirow{5}{*}{8} & \multirow{3}{*}{None} & PipeDream & 16    & 2     & 44    & 64 \\
          &       & vPipe-AS & 16    & 2     & 44    & 32 \\
          &       & \textbf{DPiper-AS} & \textbf{40} & \textbf{5} & \textbf{84} & \textbf{128} \\
\cmidrule{2-7}          & \multirow{2}{*}{S+R} & vPipe-AS & 58    & 16    & 144   & 214 \\
          &       & \textbf{DPiper-AS} & \textbf{100} & \textbf{22} & \textbf{248} & \textbf{528} \\
    \bottomrule
    \end{tabular}}
\end{table}

\subsubsection{ASynchronous Pipeline Parallel Training}
In this section, we evaluate memory efficiency in the asynchronous pipeline parallel training mode.
The results are shown in Table ~\ref{tab:maxbs-async},
where the asynchronous modes of vPipe and DawnPiper are denoted as vPipe-AS and DPiper-AS, respectively.

With memory optimization disabled,
DawnPiper significantly increases the maximum batch size compared to PipeDream and vPipe,
as GPU memory usage is more unbalanced in asynchronous pipeline parallel training.
Across the four models,
DawnPiper increases the maximum batch size by an average of 2.14$\times$ and 2.73$\times$ compared to vPipe
when the pipeline stages are 4 and 8, respectively.
Compared to PipeDream, the improvement ranges from 4.8$\times$ to 11$\times$.
This improvement is due to DawnPiper’s ability to refine
pipeline partitioning and fully utilize available GPU memory resources.
For AmoebaNet specifically, the increases are 3.13$\times$ and 4$\times$, respectively,
as vPipe is primarily designed for Transformer-based models and does not support CNN models effectively.
With memory optimization enabled, DawnPiper still achieves a 1.61$\times$ and 1.82$\times$
increase in maximum batch size compared to vPipe when the pipeline stages are 4 and 8, respectively.

\subsection{Training Performance}
In this section, we evaluate the training speed of each method under the same batch size.
% Since batch size in asynchronous pipeline parallel frameworks is not well-defined,
% we use the micro-batch size for performance comparisons.
Note that ZeRO and PipeDream can only operate with small micro-batch sizes before encountering memory oversubscription,
so they are excluded from this comparison.
It is worth noting that the time required to determine the pipeline partition and memory optimization strategy
is less than 1 second for all four models in our experiment.
% Nonetheless, we can still gain a general understanding of their performance.
% For instance, ZeRO achieves 1.4$\times$ the training speed of GPipe
% on the relatively small BERT model at a smaller batch size,
% but its speed drops to just 20\% of GPipe's on the larger GPT-2 model.
% As for PipeDream, its performance is similar to vPipe,
% with the main distinction being a more noticeable difference in model partitioning on AmoebaNet,
% while differences in other cases are relatively minor.
\begin{figure}[tb]
  \centering
  % \hspace{-0.7cm}
  \subfigure[BERT]{
    \centering
    \label{subfig:bert-4}
    \includegraphics[width=0.47\linewidth]{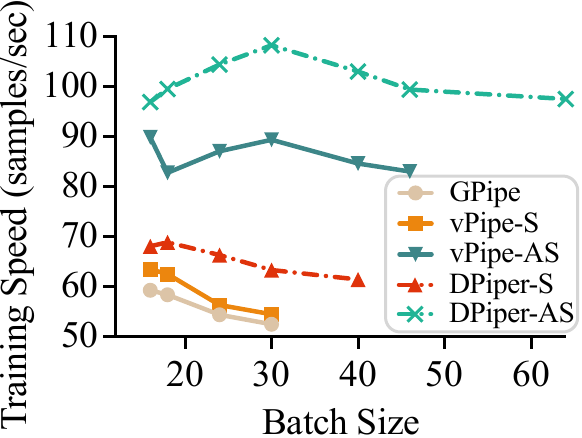}
  }
  \subfigure[GPT-2]{
    \centering
    \label{subfig:gpt-2-4}
    \includegraphics[width=0.47\linewidth]{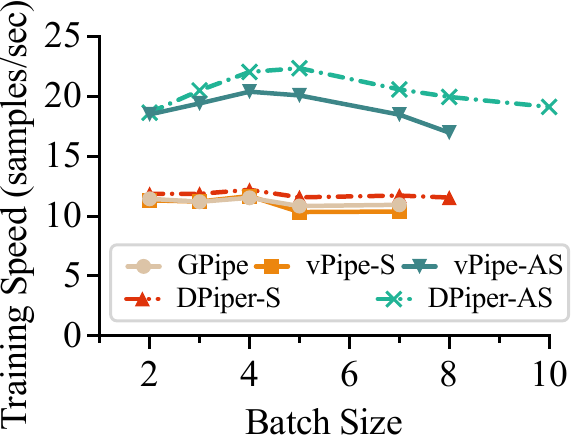}
  }
  \subfigure[T5]{
    \centering
    \label{subfig:t5-4}
    \includegraphics[width=0.47\linewidth]{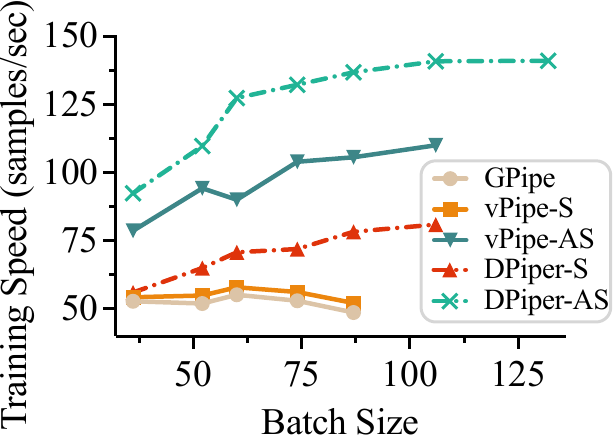}
  }
  \subfigure[AmoebaNet]{
    \centering
    \label{subfig:amoebanet-4}
    \includegraphics[width=0.47\linewidth]{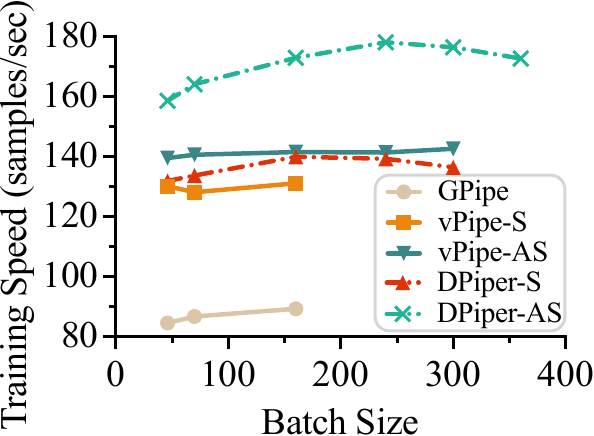}
  }
  \caption{Training Speed under Pipeline Stage of 4}
  \label{fig:perf-4}
\end{figure}

The training speed comparison for a pipeline stage of 4 is shown in Figure~\ref{fig:perf-4}.
The performance of synchronous pipeline parallelism is significantly
lower than that of asynchronous pipeline parallelism due to pipeline bubbles caused by computation synchronization.
Overall, vPipe-S outperforms GPipe, with a more noticeable advantage on AmoebaNet.
This is because GPipe does not account for communication overhead between stages,
leading to inefficient partitions in CNN models and higher communication volumes.
In contrast, both DawnPiper and vPipe take communication overhead into account when determining pipeline partitioning.

% At smaller batch sizes, DawnPiper does not exhibit a significant advantage over other methods,
% with an average performance improvement of around 7\%,
% and only a 17\% and 14\% improvement in asynchronous training on the T5 and AmoebaNet models, respectively.
% This is because, when memory is not oversubscribed,
% DawnPiper’s main benefit comes from refining the pipeline partition to find a better model split.
At smaller batch sizes, DawnPiper does not exhibit a significant advantage over other methods,
with an average performance improvement of around 10\%.
As the batch size increases, other methods must implement memory optimization to meet the training requirements.
At this point, DawnPiper shows a clear increase in training speed compared to the others.
This improvement stems from two factors. Firstly, GPU resources are better utilized as the batch size grows.
Secondly, DawnPiper can: (a) adjust the memory footprint of each stage with minimal impact on computation time,
and (b) reduce the memory footprint with negligible performance overhead using memory swapping for relatively small batch sizes.
The first advantage is possible due to insights on memory usage during training, as discussed in Section~\ref{sec:mot}.
Consequently, in synchronous pipeline parallelism,
DawnPiper achieves up to a 1.5$\times$ performance improvement on T5 as the batch size increases
and up to a 1.41$\times$ improvement in asynchronous mode.
Overall, as the batch size grows,
DawnPiper delivers an average performance improvement of 1.2$\times$, 1.12$\times$, 1.32$\times$, and 1.24$\times$ over vPipe
in asynchronous parallel training on BERT, GPT-2, T5, and AmoebaNet, respectively.
\begin{figure}
  \centering
  % \hspace{-0.7cm}
  \subfigure[GPT-2]{
    \centering
    \label{subfig:gpt-2-8}
    \includegraphics[width=0.47\linewidth]{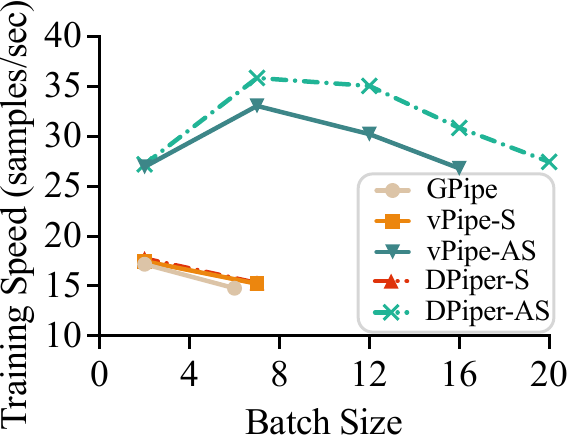}
  }
  \subfigure[T5]{
    \centering
    \label{subfig:t5-8}
    \includegraphics[width=0.47\linewidth]{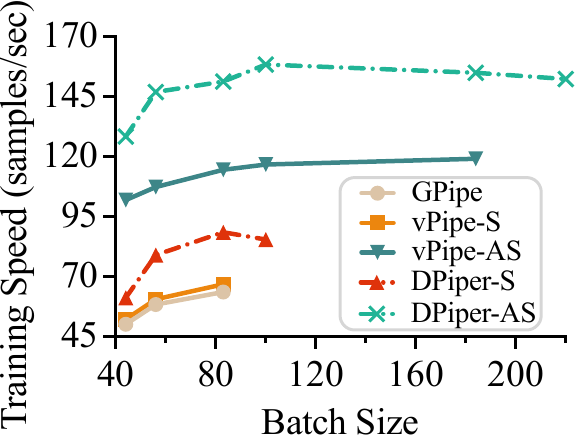}
  }
  \caption{Training Speed under Pipeline Stage of 8}
  \label{fig:perf-8}
\end{figure}

Evaluations were conducted only on GPT-2 and T5 when the number of pipeline stages increased to 8,
as the model sizes of BERT and AmoebaNet are too small to fully utilize 8 GPUs.
The results are presented in Figure~\ref{fig:perf-8}.
In synchronous mode, the performance of GPT-2 on DawnPiper, GPipe, and vPipe is comparable.
This similarity arises because the optimal partition position for computation and memory in GPT-2 are closely aligned,
limiting the space for partition optimization.
In asynchronous mode, as the batch size gradually increases,
DawnPiper's training speed is on average 1.15$\times$ faster than vPipe on GPT-2 and 1.34$\times$ faster on T5.
DawnPiper shows greater performance improvement on T5 due to its mixed architecture,
which combines both encoder and decoder,
offering a more complex structure that provides greater partitioning opportunities.
Furthermore, as the number of pipeline stages increases to 8, DawnPiper demonstrates strong scalability.

\subsection{Detailed analysis on computation time and memory footprint}
To provide a more detailed comparison of the memory footprint and computation performance between DawnPiper and vPipe,
we profiled the GPU peak memory usage and computation time of each stage on T5 in asynchronous mode with 8 pipeline stages.
The micro-batch size was set to 110, which is 2.5$\times$ larger
than the maximum batch size vPipe can handle without memory optimization.
The results are presented in Figure~\ref{fig:mem-comp}.

Regarding peak memory usage per stage, DawnPiper uses more GPU memory than vPipe,
but with more balanced memory distribution across stages.
This is because vPipe applies a coarser-grained pipeline partitioning and memory optimization,
whereas DawnPiper implements finer-grained memory optimizations tailored to each stage’s requirements.
As a result, overall GPU memory utilization is improved by 17.8\% compared to vPipe.
Additionally, DawnPiper’s partition strategy leads to more balanced computation times across stages,
with the difference between the longest and shortest stage execution times being only 7.7\%,
compared to 35.8\% under vPipe’s partition and memory optimization strategy.
\begin{figure}
  \centering
  \includegraphics*[width=0.8\linewidth]{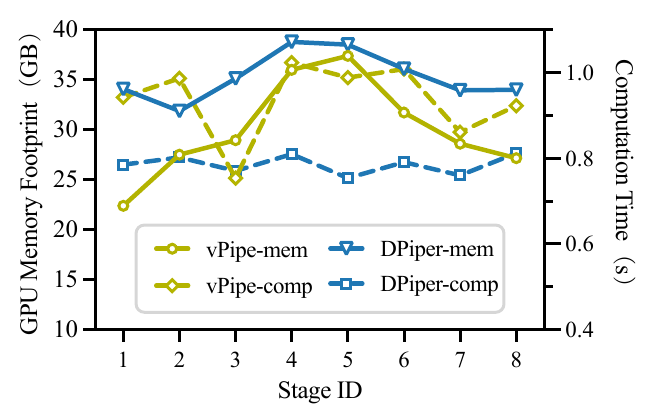}
  \caption{Peak Memory Usage and Computation time Analysis}
  \label{fig:mem-comp}
\end{figure}

%!TEX root = main.tex
% \section{Conclusions and Future Works}
\section{Conclusions}
\label{sec:conclusion}
This paper introduces DawnPiper,
a memory-scalable pipeline parallel training framework
designed to achieve efficient pipeline parallelism while minimizing GPU memory resource waste.
The framework begins by compiling and profiling the model to
obtain a fine-grained computation graph and detailed memory usage characteristics,
which refines the model partitioning and memory optimization space.
Based on observed memory usage features during training,
DawnPiper employs a binary pipeline parallel partitioning method with partition range limitations,
combined with a cost-model based memory optimization method.
% This approach significantly reduces the search space for pipeline parallel partitioning.
Experimental results on four typical DNNs demonstrate that
DawnPiper increases the trainable batch size by an average of 1.29$\times$ and 1.71$\times$ compared to GPipe and vPipe, respectively,
and up to 4.8$\times$–11$\times$ compared to PipeDream.
Additionally, it accelerates the parallel training speed of vPipe by up to 1.5$\times$.

\section*{Impact Statement}
This paper presents work whose goal is to advance the fields of pipeline parallelism in deep learning.
There are many potential societal consequences of our work, none which we feel must be specifically highlighted here.

%%%%%%%%%%%%%%%%%%%%%%%%%%%%%%%%%%%%%%%%%%%%%%%%%%%%%%%%%%%%%%%%%%%%%%%%%%%%

\bibliography{./bib/papers}
\bibliographystyle{icml2025}

%{\footnotesize
%\bibliographystyle{IEEEtran}
%\bibliography{./bib/papers}
%}

%!TEX root = main.tex
%%%%%%%%%%%%%%%%%%%%%%%%%%%%%%%%%%%%%%%%%%%%%%%%%%%%%%%%%%%%%%%%%%%%%%%%%%%%%%%
%%%%%%%%%%%%%%%%%%%%%%%%%%%%%%%%%%%%%%%%%%%%%%%%%%%%%%%%%%%%%%%%%%%%%%%%%%%%%%%
% APPENDIX
%%%%%%%%%%%%%%%%%%%%%%%%%%%%%%%%%%%%%%%%%%%%%%%%%%%%%%%%%%%%%%%%%%%%%%%%%%%%%%%
%%%%%%%%%%%%%%%%%%%%%%%%%%%%%%%%%%%%%%%%%%%%%%%%%%%%%%%%%%%%%%%%%%%%%%%%%%%%%%%

\newpage
\appendix
\onecolumn
\section{Performance Comparison Between Memory-Balanced and Compute-Balanced Partitioning}
\label{sec:pcbmbcb}
In Section~\ref{subsec:mipp},
% Based on the discussion above,
% it is evident that focusing on balancing computation across pipeline stages
% often leads to imbalanced memory usage, and vice versa.
% Additionally, with memory optimization techniques like swapping and recomputation,
we have discussed a key challenge that arises when GPU memory is oversubscribed after partitioning a model for computational balance:
should the model be repartitioned to achieve more balanced memory usage,
or should memory optimization methods be applied to handle stages that exceed GPU capacity?
Determining which approach yields better training performance is not straightforward.
% It’s difficult to determine which approach offers better training performance.
To investigate this, we implemented two simple partitioning strategies in PipeDream for comparison.
The first strategy is memory-balanced partitioning (Mem-ba),
while the second combines compute-balanced partitioning with recomputation (Comp-ba+RP).
Both strategies achieve comparable maximum batch sizes, with negligible communication overhead.
% and communication overhead does not significantly impact training performance.
% and communication overhead is negligible.
\begin{figure*}[h]
    \centering
    \includegraphics[width=0.35\linewidth]{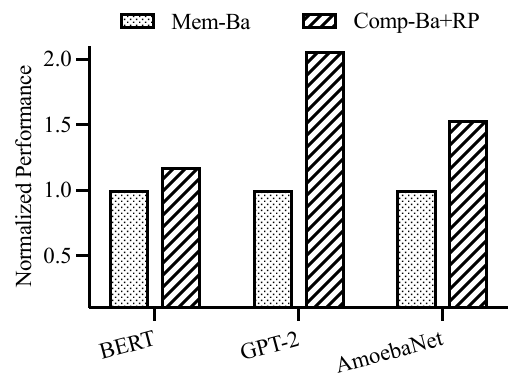}
    \caption{Normalized Performance of Memory-balanced and Compute-balanced Partitioning with Recomputation}
    \label{fig:norm-perf}
\end{figure*}

The normalized training performance results are shown in Figure~\ref{fig:norm-perf}.
Across all three models, the second method consistently outperforms the first,
particularly with the GPT-2 model,
where the training performance of Comp-Ba+RP is more than twice that of memory-balanced partitioning.
This performance gap is due to the significant difference in computation time between
the longest and shortest stages after memory-balanced partitioning in GPT-2.
Specifically, the computation time of the fourth stage reaches 286.275 ms
while the first stage only takes 29.104 ms — a nearly 10$\times$ difference.
As a result, the overall pipeline performance is severely constrained by the long computation time of the fourth stage.
Additionally, recomputation in GPT-2 provides significant memory savings with minimal additional computation overhead.
In the Comp-Ba+RP method, the total computation time for the fourth stage is reduced to 145.89 ms,
while the first stage, which requires more recomputation, increases to 136.89 ms.
This results in the fourth stage’s computation time being nearly halved compared to Mem-ba.
Similar improvements are observed in BERT and AmoebaNet, with performance gains of 1.18$\times$ and 1.54$\times$, respectively.
It is important to note that while the second method consistently outperforms in this benchmark,
this may not always hold true in every scenario.
The purpose of this micro-benchmark is to demonstrate that optimizing
pipeline parallel performance while considering memory limitations,
partitioning strategies, and memory optimization techniques is a complex and challenging problem.

\section{Supplementary Material for Section~\ref{sec:design}}
\label{sec:sm4design}
In this section, we provide additional design details of DawnPiper for Section~\ref{sec:design}.
\subsection{Memory and Compute Balanced Partitioning Algorithm}
\label{subsec:mcbpa}
In this section, we describe the process for determining the compute-balanced and memory-balanced partition position.
The general approach involves traversing the computation graph in execution order,
while continuously accumulating the forward/backward computation time or memory consumption.
Since asynchronous pipeline parallelism, such as in PipeDream,
requires accounting for stage-specific memory footprints
(due to differing numbers of activation and parameter copies across stages),
while computation time can be simply summed,
we focus on presenting the memory-balance partition algorithm for PipeDream, shown in Algorithm~\ref{alg:mem-ba}.
The other compute-balanced and memory-balanced partition can be derived by adjusting Algorithm~\ref{alg:mem-ba}.

Before executing the algorithm, we first traverse the computation graph to determine the depth of each node.
Nodes with the same depth are then sorted based on the start time of forward computation,
forming the node set $\mathcal{N}_\mathcal{G}$, which serves as an input to the algorithm.
Additional inputs include the GPU peak memory limit $M_\mathcal{G}$ and the number of pipeline stages $\ell$.
The algorithm outputs the partition positions $\rho_x$ for each stage.
In PipeDream, the ratio of activation and parameter copies that need to be stored across stages is
$\ell:(\ell-1):(\ell-2):\cdots:1$.
If the peak memory usage of a micro-batch for each stage is denoted as $M_x$,
the following equation must hold to ensure balanced peak memory usage across all stages.
\begin{align}
  \ell \times M_1=\cdots=(\ell-x+1) \times M_x=\cdots=M_{\ell}
\end{align}
Based on this equation, the $M_x$ for each stage is computed,
as outlined in lines 2-5 and line 12 of Algorithm~\ref{alg:mem-ba}.
The algorithm then iteratively traverses the nodes,
first adding the activation memory, parameter memory,
and optimizer states memory of each node,
while checking if the accumulated memory exceeds the previously recorded peak memory.
Next, it subtracts the memory that needs to be released at the current node,
repeating this process until the accumulated peak memory reaches the target $M_x$ for the current stage.
At this point, the current node marks the partition position for that stage.
Although the description in the algorithm simplifies the partition position as a single node,
it typically involves a set of nodes in practice.
\begin{algorithm}
    \caption{Memory balance partition in PipeDream}
    \centering
    \label{alg:mem-ba}
\begin{algorithmic}[1]
    \INPUT nodes collection $N_\mathcal{G}$, peak memory usage $M_\mathcal{G}$, pipeline depth $\ell$
    \OUTPUT $\bigcup_{1 \leqslant x \leqslant \ell}\rho_x$
    \STATE $x \leftarrow 1,cur\_mem \leftarrow 0,peak\_mem \leftarrow 0$
    \FOR{$i$ in $[0,\ell)$}
    \STATE $sum += \ell \div (\ell-i)$
    \ENDFOR
    \STATE $M_x \leftarrow M_\mathcal{G} \div sum$
    \FOR{$n$ in $N_\mathcal{G}$}
        \STATE $cur\_mem += m_a^n + m_p^n$
        \STATE $peak\_mem \leftarrow \max(cur\_mem,peak\_mem)$
        \STATE $cur\_mem -= m_d^n$
        \IF{$peak\_mem \ge M_x$}
        \STATE $\rho_x \leftarrow n,x \leftarrow x + 1$
        \STATE $M_x \leftarrow (\ell \div (\ell -x + 1) \times M_1)$
        \STATE $cur\_mem, peak\_mem \leftarrow 0$
        \IF{$x == \ell$}
            \STATE \textbf{break}
        \ENDIF
        \ENDIF
    \ENDFOR
\end{algorithmic}
\end{algorithm}

\subsection{Communication Optimization}
\label{subsec:commopt}
A key assumption in Theorem~\ref{thm:ppp} is
that communication overhead should not impact the efficiency of pipeline parallelism.
Therefore, it is essential to avoid partition points that generate significant
communication volumes during the traversal within the partition range.
As discussed in Section~\ref{sec:mot}, the majority of computation nodes have very small activation memory.
First, we can identify inevitable communication points,
i.e., nodes with edges extending far outside
the partition range where data transmission is unavoidable regardless of the partition strategy,
such as the output of the initial embedding layer in BERT.
This inevitable communication can be disregarded.
For the remaining nodes, communication overhead can be minimized by optimizing the partition positions.
For instance, if activations from multiple nodes at the same depth need to be transmitted to the next stage,
the algorithm can identify the lowest common leaf node connecting these edges.
If this leaf node is the sole connection point,
the partition can be adjusted so that only the activation of this single node is transmitted, instead of multiple nodes.
Such leaf nodes are typically easy to locate,
resulting in minimal impact on both the computation time and memory usage of the adjacent stages.

\section{Implementation Details}
\label{sec:impl}
In this section, we provide additional detail on the implementation of the proposed DawnPiper,
which may not be thoroughly specified in the main paper due to the page limit.
This includes the process for compiling and profiling a model in PyTorch,
as well as the implementation of memory swapping and recomputation.
Although these implementations are tailored to PyTorch,
the methods proposed are broadly applicable to other deep learning frameworks,
such as TensorFlow~\cite{abadiTensorFlowSystemLargeScale2016}.

\subsection{Compilation and Profiling}
\textbf{Compilation:} The purpose of model compilation
is to generate a fine-grained computation graph,
which expands the model partitioning and memory optimization space
while enabling automatic code generation for each pipeline stage.
Unlike tensor parallelism, this process does not involve splitting and tracking operations within individual operators.
To achieve this, we leverage the \emph{torch.fx} library~\cite{reedTorchFxPractical2022},
a pure Python toolkit provided by PyTorch for capturing and transforming neural network structures via tracing.
\emph{torch.fx} traces the model defined by user scripts without
dissecting internal computations within \texttt{nn.Module} operators.
Users can modify the captured computation graph freely,
and \emph{torch.fx} can then generate the corresponding Python code based on the updated graph.

\textbf{Profiling:} To profile the forward and backward computation times of nodes after compilation,
we insert timing hook functions at the start and end of the respective computations.
The initial training iterations during the warmup phase are excluded,
and the average is calculated over 50 batches of stable iterations.
For capturing the memory sizes of activations, parameters, and optimizer states,
we pass proxy input tensors of equivalent size through the computation graph and execute it once.
The corresponding memory sizes are extracted from the result tensors once each node’s computation completes.
Additionally, the saved tensor information for each node is retrieved during this execution
using the \texttt{torch.autograd.graph.saved\_tensors\_hooks} function provided by PyTorch, which we register for each node.

\subsection{Swap and Recomputation}
\textbf{Memory swap:} PyTorch does not natively support memory swapping operations.
Although a tensor can be transferred to a target device using the \texttt{Tensor.to(device)} interface,
its underlying GPU memory cannot be released without deleting the tensor object after the transfer completes.
To address this, we modified the underlying source code of PyTorch’s tensor data structure,
adding two functions: one that directly releases the underlying GPU memory without deleting the tensor object,
and another that sets the underlying memory address to a specified parameter.
This modification allows the original memory to be released after swapping out,
while setting the original tensor’s memory address to that of the swapped-in tensor.
For triggering memory swapping operations,
we use PyTorch's hook registration function to insert the operations at the appropriate points before training begins.
Once set, these memory swaps are consistently executed according to the defined strategy during pipeline parallel training.

\textbf{Recomputation:} PyTorch offers the \texttt{torch.utils.checkpoint} library for
handling recomputation by rewriting the parts of the module code that require it.
This library initiates recomputation when the corresponding backward computation starts,
functioning as on-demand recomputation.
Since we do not need another GPU stream for early recomputation process,
the library’s functionality is more than adequate for our needs.
Additionally, the module code for recomputation within each stage is
automatically generated during the module code generation phase.

\end{document}